\documentclass[preprint,12pt]{elsarticle}
\usepackage{ifthen}
\newcommand{\techRep}{true} %% switch here between true and false
\newcommand{\iftechrep}{\ifthenelse{\equal{\techRep}{true}}}

\usepackage{amsmath}
\usepackage{amssymb}
\usepackage{btran}
\usepackage{tikz}

\newtheorem{theorem}{Theorem}[section]
\newtheorem{lemma}{Lemma}[section]
\newtheorem{proposition}{Proposition}[section]

\newdefinition{definition}{Definition}[section]
\newproof{proof}{Proof}

\journal{Information Processing Letters}

\begin{document}

\renewcommand{\P}[1]{{\cal P}\left(#1\right)}
\newcommand{\R}{\mathbb{R}}
\newcommand{\N}{\mathbb{N}}
\newcommand{\norm}[1]{\lVert#1\rVert}
\newcommand{\ind}{\mbox{}\hspace{7mm}}
\newcommand{\indd}{\ind\ind}
\newcommand{\inddd}{\indd\ind}
\newcommand{\indddd}{\inddd\ind}
\newcommand{\inddddd}{\indddd\ind}
\newcommand{\tran}[1]{\xrightarrow{#1}}
\newcommand{\qacc}{q_\mathit{acc}}
\newcommand{\qrej}{q_\mathit{rej}}
\newcommand{\init}{\vdash}
\newcommand{\fin}{\dashv}
\newcommand{\enc}[1]{\langle #1 \rangle}
\newcommand{\bdec}[1]{\mathit{num}(#1)}
\newcommand{\benc}[1]{[ #1 ]}
\newcommand{\Bin}{\mathit{Bin}}

\begin{frontmatter}

%% Title, authors and addresses

%% use the tnoteref command within \title for footnotes;
%% use the tnotetext command for the associated footnote;
%% use the fnref command within \author or \address for footnotes;
%% use the fntext command for the associated footnote;
%% use the corref command within \author for corresponding author footnotes;
%% use the cortext command for the associated footnote;
%% use the ead command for the email address,
%% and the form \ead[url] for the home page:
%%
%% \title{Title\tnoteref{label1}}
%% \tnotetext[label1]{}
%% \author{Name\corref{cor1}\fnref{label2}}
%% \ead{email address}
%% \ead[url]{home page}
%% \fntext[label2]{}
%% \cortext[cor1]{}
%% \address{Address\fnref{label3}}
%% \fntext[label3]{}

\title{BPA Bisimilarity is EXPTIME-hard}

%% use optional labels to link authors explicitly to addresses:
%% \author[label1,label2]{<author name>}
%% \address[label1]{<address>}
%% \address[label2]{<address>}

\author{Stefan Kiefer\fnref{mylabel}}
\fntext[mylabel]{Stefan Kiefer is supported by the EPSRC.}

\address{University of Oxford, UK}

\begin{abstract}
 Given a basic process algebra (BPA) and two stack symbols, the BPA bisimilarity problem asks whether the two stack symbols are bisimilar.
 We show that this problem is EXPTIME-hard.
\end{abstract}

\begin{keyword}
%% keywords here, in the form: keyword \sep keyword
basic process algebra \sep bisimilarity \sep computational complexity

%% MSC codes here, in the form: \MSC code \sep code
%% or \MSC[2008] code \sep code (2000 is the default)

\end{keyword}

\end{frontmatter}

\section{Introduction} \label{sec-intro}

Equivalence checking is the problem of determining whether two systems are semantically identical.
This is an important question in automated verification and, more generally,
 represents a line of research that can be traced back to the inception of theoretical computer science.
In particular, \emph{bisimilarity} is a fundamental notion for process algebraic formalisms~\cite{Milner}
 and enjoys pleasant mathematical properties.
As a result, a great deal of research in the analysis of infinite-state processes
 (such as pushdown automata or Petri nets)
 has been devoted to deciding bisimilarity of two given processes,
  see e.g.~\cite{KuceraJ06} for a comprehensive overview.

In this note we study bisimilarity for \emph{basic process algebras (BPAs)}.
A BPA consists of rules of the form
$
 X \btran{a} Y_1 Y_2 \cdots Y_k
$,
where $X, Y_1, \ldots, Y_k \in \Gamma$ are \emph{stack symbols} and $a$ is an \emph{action}, see Section~\ref{sec-prel} for formal details.
A BPA induces a (generally infinite) labelled transition system on~$\Gamma^*$ with $X \beta \tran{a} \alpha \beta$ whenever $X \btran{a} \alpha$
 (where $\alpha,\beta \in \Gamma^*$).
BPAs are also called \emph{context-free processes} and are closely related to context-free grammars in Greibach normal form.
They can be viewed as \emph{pushdown automata} with a single control state.

There is a large body of literature on decidability and complexity of determining bisimilarity for
 two processes in the \emph{process rewrite systems} hierarchy and beyond,
 see~\cite{KuceraJ06,SrbaRoadmap} for comprehensive and up-to-date surveys.
For many process classes the precise complexity of bisimilarity is unknown.
Let us focus here on the results that pertain most to the present note.
Although \emph{language} equivalence for context-free grammars and hence for BPAs is undecidable,
 bisimilarity of BPAs was shown to be decidable~\cite{Christensen92} in doubly exponential time~\cite{Burkart95,Burkart01,JancarBPA13}.
S{\'e}nizergues proved decidability of bisimilarity for general pushdown automata,
 even for the slightly larger class of \emph{equational graphs of finite out-degree}~\cite{Senizergues05}.
Lower complexity bounds have also been obtained:
St\v{r}\'{\i}brn\'{a}~\cite{Stribrna98} showed that \emph{weak} bisimilarity of BPAs is PSPACE-hard.%
\footnote{
``Weak'' in the context of bisimilarity means that transitions may be labelled by a ``non-visible'' action~$\tau$ that can be combined with any other action.
We mean strong bisimilarity in this note unless explicitly stated otherwise.
}
Mayr~\cite{Mayr00} subsequently proved PSPACE-hardness for bisimilarity of pushdown automata.
Srba~\cite{SrbaPSPACEhardBPA} improved both results by showing PSPACE-hardness for bisimilarity of BPAs.
Ku\v{c}era and Mayr~\cite{KuceraMayr02} proved that bisimilarity of pushdown automata is EXPTIME-hard.
Then Mayr~\cite{Mayr2005} showed EXPTIME-hardness also for weak bisimilarity of BPAs.
In this note we prove that bisimilarity of BPAs is EXPTIME-hard,
 thus improving or subsuming all hardness results mentioned above.

Our result establishes a contrast between BPAs and two related models,
 called \emph{one-counter processes} and \emph{basic parallel processes (BPPs)}, respectively.
One-counter processes are another subclass of pushdown automata with only two stack symbols, one of which is a bottom-of-stack marker.
Bisimilarity for one-counter processes is PSPACE-complete~\cite{BohmGJ10}.
BPPs can be viewed as a parallel (or commutative) variant of BPAs in which the rules may not only rewrite the leftmost stack symbol but an arbitrary one.
BPP bisimilarity was shown PSPACE-complete as well~\cite{SrbaPSPACEhardBPA,Jancar03}.
Our EXPTIME lower bound shows that the complexity for BPAs is different (unless EXPTIME = PSPACE).

The following known results rule out certain possibilities to extend our new lower bound:
EXPTIME-hardness also holds for weak bisimilarity of \emph{normed} BPAs~\cite{Mayr2005},
 where ``normed'' means that the labelled transition system has from each $X \in \Gamma$ a path to the empty word.
However, strong bisimilarity of normed BPAs is decidable in polynomial time~\cite{Hirshfeld1996},
 see also \cite{CzerwinskiL10} for a recent development.
Similarly, bisimilarity of \emph{visibly} pushdown automata is EXPTIME-complete~\cite{Srba09},
 whereas bisimilarity of visibly BPAs is in~P~\cite{Srba09}.

From a technical point of view, the proof in this note improves Srba's PSPACE-hardness proof~\cite{SrbaPSPACEhardBPA}.
A careful inspection shows that his reduction (from the QBF satisfiability problem to bisimilarity of BPAs) can be decomposed into two parts:
\begin{itemize}
\item The first part is a reduction to a particular reachability game on a \emph{counter},
 where one of the players attempts to reach a configuration from a particular set, and the other player wants to avoid that.
 The counter is \emph{succinct}, i.e., the numbers are represented in binary.
\item
  The second part is a reduction from this game to BPA bisimilarity.
\end{itemize}
Thus, the counter game was implicitly proved PSPACE-hard in~\cite{SrbaPSPACEhardBPA}.
In this note we consider the counter game separately and show that it is in fact EXPTIME-complete.
This is done by adapting (or using) an EXPTIME-completeness result on succinct counters from~\cite{JurdzinskiSL08}.
The second part of Srba's reduction (from the counter game to BPA bisimilarity) is then easily adapted.

\section{Preliminaries} \label{sec-prel}

Let $\N$ denote the set of nonnegative integers.

A \emph{labelled transition system (LTS)} is a triple $(S,\Sigma,\mathord{\tran{}})$,
 where $S$ is a countable set of \emph{states},
 $\Sigma$ is a finite set of \emph{actions}, and
 $\mathord{\tran{}} \subseteq S \times \Sigma \times S$ is a \emph{transition relation}.
We write $s \tran{a} t$ to mean $(s,a,t) \in \mathord{\tran{}}$.
The LTS is \emph{finitely branching} if each $s \in S$ has only finitely many outgoing transitions $s \tran{a} t$.

Given an LTS, a \emph{(strong) bisimulation} is a relation $R \subseteq S \times S$ such that for all $(s,s') \in R$ and $a \in \Sigma$ we have:
 (1) for all $t$ with $s \tran{a} t$ there is $t'$ with $s' \tran{a} t'$ and $(t,t') \in R$; and
 (2) for all $t'$ with $s' \tran{a} t'$ there is $t$ with $s \tran{a} t$ and $(t,t') \in R$.
We write $s \sim s'$ and say that $s$ and~$s'$ are \emph{(strongly) bisimilar} if there is a bisimulation~$R$ with $(s,s') \in R$.
We remark that $\mathord{\sim}$ is an equivalence relation and the the union of all bisimulations.
Due to K\H{o}nig's lemma we also have the following inductive characterisation of bisimilarity of finitely branching LTSs:
Consider the decreasing sequence of equivalence relations
 $\mathord{\sim_0} \supseteq \mathord{\sim_1} \supseteq \mathord{\sim_2} \supseteq \cdots$
 defined as the equivalence relations with $\mathord{\sim_0} := S \times S$,
 and $s \sim_{\ell+1} s'$ if for all $s \tran{a} t$ there is $s' \tran{a} t'$ with $t \sim_\ell t'$.
Then we have $\mathord{\sim} = \bigcap_{\ell \in \N} \mathord{\sim_\ell}$;
 so $s \not\sim s'$ implies $s \not\sim_{\ell} s'$ for some $\ell \in \N$.

A \emph{basic process algebra (BPA)} is a triple $(\Gamma,\Sigma,\mathord{\btran{}})$,
 where $\Gamma$ is a finite set of \emph{stack symbols},
 $\Sigma$ is a finite set of \emph{actions},
 and $\mathord{\btran{}} \subseteq \Gamma \times \Sigma \times \Gamma^*$ is a finite set of \emph{transition rules}.
We write $X \btran{a} \alpha$ to mean $(X,a,\alpha) \in \mathord{\btran{}}$.
A stack symbol~$X \in \Gamma$ is \emph{dead} if there is no outgoing rule $X \btran{a} \alpha$.
A BPA generates a finitely branching LTS $(\Gamma^*,\Sigma,\mathord{\tran{}})$ with $X \beta \tran{a} \alpha \beta$ if $X \btran{a} \alpha$
 (where $X \in \Gamma$ and $\alpha,\beta \in \Gamma^*$).

\section{Main Result} \label{sec-main}

The \emph{BPA bisimilarity problem} asks, given a BPA $(\Gamma,\Sigma,\mathord{\btran{}})$ and two stack symbols $X,X' \in \Gamma$ whether $X \sim X'$
 holds in the generated LTS.
We prove the following theorem:
\begin{theorem} \label{thm-main}
 The BPA bisimilarity problem is EXPTIME-hard, even if the BPA has no dead stack symbols and only two actions.
 It is also EXPTIME-hard if the BPA has only one dead symbol and only one action.
\end{theorem}

We prove Theorem~\ref{thm-main} in two steps.
In Section~\ref{sub-counter} we show that determining the winner in a particular reachability game on a counter is EXPTIME-hard,
 if the counter is succinct, i.e., the numbers are written in binary.
The proof follows closely a proof from~\cite{JurdzinskiSL08}.
In Section~\ref{sub-bisim} we reduce the succinct counter problem to the BPA bisimilarity problem.

\subsection{A Reachability Game On a Succinct Counter} \label{sub-counter}

A \emph{hit-or-run game} is a tuple $(S_0, S_1, \mathord{\tran{}}, s_\init, s_\fin, k_\fin)$ where
 $S = S_0 \cup S_1$ for disjoint $S_0, S_1$ is a finite set of states,
  and $\mathord{\tran{}} \subseteq S \times \N \times (S \cup \{s_\fin\})$ is a transition relation,
   $s_\init \in S$ is the initial state,
   $s_\fin \not\in S$ is the final state,
    and $k_\fin \in \N$ is the final value.
We write $s \tran{\ell} t$ if $(s, \ell, t) \in \mathord{\tran{}}$.
We require that for each $s \in S$ there is at least one outgoing transition $s \tran{\ell} t$.
A \emph{configuration} of the game is a pair $(s,k) \in (S \cup \{s_\fin\}) \times \N$.
The game is played by two players, Player~$0$ (``she'') and Player~$1$ (``he'').
The game starts in configuration $(s_\init,0)$ and proceeds in \emph{moves}:
 if the current configuration is $(s,k) \in S_i \times \N$ for $i \in \{ 0,1 \}$, then Player~$i$ chooses a transition $s \tran{\ell} t$.
The resulting new configuration is $(t, k + \ell)$.
Player~$1$'s goal is to reach a configuration $(s_\fin,k)$ with $k \ne k_\fin$.
Consequently, Player~$0$'s goal is to keep the game within $\{(s_\fin,k_\fin)\}\ \cup \ S \times \N$.
The name ``hit-or-run'' refers to Player~$0$'s options to win the game: ``hit'' $(s_\fin,k_\fin)$ or ``run'' from~$s_\fin$.

\begin{proposition} \label{prop-counter-game}
 Given a hit-or-run game with numbers written in binary, the problem of determining the winner is EXPTIME-complete.
\end{proposition}
\begin{proof}
For the upper bound recall that Player~$1$'s objective is to reach any configuration in a given set.
So if he can win, he can win using a positional strategy.
Moreover, the counter value never decreases and Player~$1$ cannot benefit from repeating configurations.
So if he can win, he can win within exponentially many steps.
It follows that one can construct a straightforward alternating PSPACE Turing machine
 that accepts if and only if Player~$1$ can win:
 the existential steps correspond to Player-$1$ moves, the universal steps correspond to Player-$0$ moves.
The machine uses an extra counter (in polynomial space) to reject if it does not accept within an exponential time bound.
Therefore the problem is in APSPACE, which equals EXPTIME.

For the lower bound we adapt a proof given in~\cite{JurdzinskiSL08} for so-called \emph{countdown games}.
We give a polynomial-time reduction from the problem of acceptance of a word by a PSPACE-bounded alternating Turing machine.
Let $M = (\Sigma, Q_\exists, Q_\forall, \delta, q_\init, \qacc, \qrej)$ be a PSPACE-bounded alternating Turing machine,
 where $\Sigma$ is a finite alphabet, and $Q = Q_\exists \cup Q_\forall$ is a finite set of (control) states partitioned into existential states~$Q_\exists$
  and universal states~$Q_\forall$,
 and $\delta \subseteq Q \times \Sigma \times Q \times \Sigma \times \{L, R\}$ is a transition relation.
A transition $(q, a, q', a', D) \in \delta$ means that if $M$ is in state $q$ and its head reads letter~$a$,
 then it rewrites the contents of the current cell with the letter~$a'$, it moves the head in direction~$D$ (either left if $D=L$, or right if $D=R$),
  and it changes its state to~$q'$.
We assume that for all $q \in Q$ and $a \in \Sigma$ there is at least one outgoing transition,
 and that $M$ does self-loops in~$\qacc$ and in~$\qrej$.
%We assume that $\qacc$ and~$\qrej$ do not have any outgoing transitions and that
% for all $q \in Q \setminus\{\qacc, \qrej\}$ and $a \in \Sigma$ there is at least one outgoing transition.
W.l.o.g.\ we can also assume that all computations reach either $\qacc$ or~$\qrej$,
 and no configuration is repeated before that (this is achieved, e.g., using a counter on the tape).

Let $w \in \Sigma^n$ be the input word.
We can assume that during its computation $M$ uses exactly $N$ tape cells (with $N$ polynomial in~$n$),
 so we can encode a tape content as word $u \in \Sigma^N$.
Let $G := |\Sigma|$.
Let $\enc{\cdot} : \Sigma \to \{0, 1, \ldots, G-1\}$ be a bijection.
For every~$a \in \Sigma$, it is convenient to think of~$\enc{a}$ as a $G$-ary digit,
 and we extend~$\enc{\cdot}$ to $\Sigma^N$ by $\enc{a_{0} a_{1} \cdots a_{N-1}} := \sum_{i=0}^{N-1} \enc{a_i} \cdot G^i < G^N$.
In this way every tape content in $\Sigma^N$ can be seen as a residue class modulo~$G^N$,
 and a rewrite of the tape can be simulated by adding a number.

%\begin{tikzpicture}[scale=1.5]
%  \node (N1) at (0,0) {$s^N_1$};
%  \node (N2) at (1,0) {$s^N_2$};
%  \draw[->] (N1) to[bend left] node[above] {$G^N$} (N2);
%  \draw[->] (N1) to[bend right] node[below] {$0$} (N2);
%  \node at (1.5,0) {$\cdots$};
%  \node (NG1) at (2,0) {$s^N_{G-1}$};
%  \node (N11) at (3,0) {$s^{N+1}_1$};
%\end{tikzpicture}

We define a hit-or-run game so that Player~$0$ can win it if and only if $M$ accepts~$w$.
The main part of the game is constructed so that a play simulates a computation of~$M$ on~$w$.
The (control) states of the game encode the (control) states of~$M$ and the position of the read head.
In each step, Player~$0$ will make a claim about what the head currently reads.
Then Player~$1$ is given the option to doubt this claim.
If he doubts, the claim is checked: Player~$0$ will win the game if and only if her claim was true.
If Player~$1$ does not doubt Player~$0$'s claim, Player~$0$ or Player~$1$ will pick a transition
 if the simulation is currently in an existential or universal state, respectively.
The tape rewrite can be simulated by adding a suitable number,
 which only depends on the position~$i$ of the head and the old and the new content of the $i$th tape cell.
If the simulation reaches $\qacc$ or~$\qrej$, then Player~$0$ or Player~$1$, respectively, will win the game.
Therefore, if $M$ rejects~$w$, Player~$0$ will be forced to ``lie'' eventually, which enables Player~$1$ to win.
On the other hand, if $M$ accepts~$w$, Player~$0$ can simulate correctly until the computation has reached~$\qacc$,
 so in order to win, Player~$1$ needs to wrongfully doubt a correct claim made by Player~$0$,
  which Player~$0$ can punish by winning the game.

We now give the details.
For every $q \in Q$ and $i \in \{0, \ldots, N-1\}$, the game includes a state~$(q,i) \in S_0$ (exception: $(\qrej,i) \in S_1$).
A game configuration $((q,i),k)$ corresponds to the configuration of~$M$ with state~$q$,
 the head at position~$i$, and tape content $k \bmod G^N$.
For each $(q,i,a) \in Q \times \{0, \ldots, N-1\} \times \Sigma$ there is a state~$(q,i,a) \in S_1$
 and a transition $(q,i) \tran{0} (q,i,a)$.
By choosing the transition to~$(q,i,a)$, Player~$0$ claims that the tape cell at position~$i$ currently contains~$a$.
If Player~$1$ accepts the claim, he takes a transition $(q,i,a) \tran{0} (q,i,a,*)$ where $(q,i,a,*) \in S_0$ if $q \in Q_\exists$
 and $(q,i,a,*) \in S_1$ if $q \in Q_\forall$.
If Player~$1$ doubts the claim, he takes a transition $(q,i,a) \tran{0} s^{i,a}_0$ where $s^{i,a}_0 \in S_0$ is a state from which the claim will be checked,
 as we describe below.
%For each $\tau = (q,a,q',a',D) \in \delta$ there is a state $(q,i,a,\tau) \in S$ and a transition $(q,i,a,*) \tran{0} (q,i,a,\tau)$,
% which is taken by Player~$1$ (or Player~$0$, respectively) to pick an outgoing transition.
%There is a transition $(q,i,a,\tau) \tran{k(i,a,a')} (q',i')$ where $i' = i-1$ if $D=L$ and $i' = i+1$ if $D=R$,
% and $k(i,a,a')$ is chosen to reflect the tape update at position~$i$ from~$a$ to~$a'$.
For each $(q,a,q',a',D) \in \delta$ there is a transition $(q,i,a,*) \tran{k(i,a,a')} (q',i')$ where $i' = i-1$ if $D=L$ and $i' = i+1$ if $D=R$,
 and $k(i,a,a')$ is chosen to reflect the tape update at position~$i$ from~$a$ to~$a'$.
This is achieved by taking $k(i,a,a') := G^i \cdot (\enc{a'} - \enc{a}) + G^N$.
In order to make $(\qacc,i) \in S_0$ winning and $(\qrej,i) \in S_1$ losing for Player~$0$,
 we add transitions $(\qacc,i) \tran{0} (\qacc,i)$ and $(\qrej,i) \tran{1} (\qrej,i)$ and $(\qrej,i) \tran{0} s_\fin$,
  where $s_\fin$ is the final state of the hit-or-run game, which will be specified later on.

It remains to define the game from the states $s^{i,a}_0$ on from which it will be checked whether the tape cell at position~$i$ contains the letter~$a$.
This will be done in two phases.
In the first phase we allow Player~$0$ to add to the counter so that the first $N$ ``$G$-ary digits'' are set to~$0$,
 except at position~$i$ where she can set it to~$0$ only if the $i$th tape cell contains~$a$.
Thus Player~$0$ can reach the end of the first phase with a multiple of~$G^N$ if and only if her claim was true.
During the simulation of the computation and also during the first phase that was just described,
 each step, apart from modifying the lower $N$ digits, increases the counter by~$G^N$.
As a consequence, the digits at positions $N, \ldots, N'$ may have a nonzero value, where $N' \in \N$ is polynomial in~$n$.
In the second phase we give Player~$0$ the possibility to set all those digits to $G-1$.
Therefore, Player~$0$ can reach counter value $k_\fin := G^{N'} - G^N$ in the final state if and only if she did not lie before entering~$s^{i,a}_0$.

For the first phase we add states $s^{i,a}_j \in S_0$ for all $j \in \{0, \ldots, N\}$,
 and transitions $s^{i,a}_j \tran{-G^j \cdot \ell + G^N} s^{i,a}_{j+1}$ for all $j \in \{0, \ldots, N-1\} \setminus \{i\}$ and $\ell \in \{0, \ldots, G-1\}$,
 and a transition $s^{i,a}_i \tran{-G^i \cdot \enc{a} + G^N} s^{i,a}_{i+1}$.
We identify $s^{i,a}_N$ with a single state~$s_N \in S_0$, which marks the end of the first phase and the start of the second phase.

For the second phase, recall that if $M$ rejects~$w$, then Player~$0$ is forced to lie after at most $m := |Q| \cdot N \cdot |\Sigma|^N$
 transitions (as there are no repeating configurations).
Each of those transitions, apart from modifying the lower $N$ digits, increases the counter by~$G^N$.
The same holds for the $N$ transitions of the first phase.
So we can take $N' := \min \{i \in \N \mid G^{i} - G^N \ge G^N \cdot (m + N)\}$.
We add states $s_{N+1}, \ldots, s_{N'-1} \in S_0$ and the final state~$s_{N'}$ and transitions $s_i \tran{G^i \cdot \ell} s_{i+1}$
 for all $i \in \{N, \ldots, N'-1\}$ and $\ell \in \{0, \ldots, G-1\}$.

The complete hit-or-run game $(S_0, S_1, \mathord{\tran{}}, s_\init, s_\fin, k_\fin)$ consists of the states and transitions described above
 and $s_\fin := s_{N'}$ and $k_\fin := G^{N'} - G^N$ and an initial state $s_\init \in S$ with a transition $s_\init \tran{\enc{u_\init}} (q_\init,0)$,
  where $u_\init$ denotes the initial tape content of~$M$ on~$w$, and we assume that the initial position of the read head is~$0$.
\qed
\end{proof}

\iftechrep{We}{In~\cite{KieferBPA-techrep} we} offer an alternative EXPTIME-hardness proof that is shorter but not completely self-contained,
 as it relies on an EXPTIME-hardness proof given in~\cite{JurdzinskiSL08} for countdown games.

\iftechrep{%
\begin{proof}[shorter proof of the lower bound in Proposition~\ref{prop-counter-game}]
We reduce from the problem of determining the winner in a countdown game~\cite{JurdzinskiSL08}.
A \emph{countdown game} is a tuple $(Q, \mathord{\ctran{}}, q_\init, k_\fin)$ where
 $Q$ is a finite set of states,
 $\mathord{\ctran{}} \subseteq Q \times \N\setminus\{0\} \times Q$ is a transition relation,
 $q_\init \in Q$ is the initial state,
 and $k_\fin$ is the final value.
We write $q \ctran{\ell} r$ if $(q, \ell, r) \in \mathord{\ctran{}}$.
A configuration of the game is an element $(q,k) \in Q \times \N$.
The game starts in configuration $(q_\init,0)$ and proceeds in moves:
 if the current configuration is $(q,k) \in Q \times \N$,
  first Player~$0$ chooses a number~$\ell$ with $0 < \ell \le k_\fin - k$ and $q \ctran{\ell} r$ for at least one $r \in Q$;
  then Player~$1$ chooses a state $r \in Q$ with $q \ctran{\ell} r$.
The resulting new configuration is $(r,k+\ell)$.
Player~$0$ wins if she hits a configuration from $Q \times \{k_\fin\}$, and she loses if she cannot move (and has not yet won).
(We have slightly paraphrased the game from~\cite{JurdzinskiSL08} for technical convenience,
 rendering the term \emph{countdown} game somewhat inept.)

The problem of determining the winner in a countdown game was shown EXPTIME-complete in~\cite{JurdzinskiSL08}.
Let $(Q, \mathord{\ctran{}}, q_\init, k_\fin)$ be a countdown game.
W.l.o.g.\ we can assume that each $q \in Q$ has an outgoing transition $q \ctran{\ell} r$
 (this can, e.g., be achieved by adding self-loops $q \ctran{k_\fin+1} q$).
% and that all occurring numbers are at most $k_\fin+1$.
We show how to compute in polynomial time a hit-or-run game $(S_0, S_1, \mathord{\tran{}}, q_\init, s_N, k_\fin)$
 so that Player~$0$ can win the countdown game if and only if she can win the hit-or-run game.
Let $N \in \N$ such that $2^{N-1} \ge k_\fin$.
We include states $s_0, \ldots, s_{N-1} \in S_0$ and transitions $s_i \tran{2^i} s_{i+1}$ and $s_i \tran{0} s_{i+1}$
 for all $i \in \{0, \ldots, N-1\}$.
Observe that in the hit-or-run game a configuration $(s_0,k)$ is winning for Player~$0$ if and only if $k \le k_\fin$.
We include $Q \subseteq S_0$, and for all $(q,\ell) \in Q \times \N$ with $q \ctran{\ell} r$ for some~$r$ we include a state $q^\ell \in S_1$.
For each transition $q \ctran{\ell} r$
 we include %a state $q \in S_0$ and a state $q^\ell \in S_1$ and
 transitions $q \tran{\ell} q^\ell$ and $q^\ell \tran{0} r$ and $q \tran{0} s_N$ and $q^\ell \tran{0} s_0$.

We show that this reduction works.
Assume that Player~$0$ can win the countdown game.
Then she can emulate her strategy in the hit-or-run game until a configuration $(q,k_\fin)$ is reached,
 which allows her to move to configuration $(s_N,k_\fin)$ and win the game.
If Player~$1$ interrupts this strategy by moving to a configuration $(s_0,k)$ for $k \le k_\fin$,
 Player~$0$ wins as well as described above.
Now assume that Player~$1$ can win the countdown game.
Then he can emulate his strategy in the hit-or-run game until
 Player~$0$ is forced to move to a configuration~$(s_N,k)$ with $k<k_\fin$ (winning for Player~$1$)
  or to a configuration $(q^\ell,k)$ with $k > k_\fin$ (winning for Player~$1$ as well, as he can then move to~$(s_0,k)$).
\qed
\end{proof}
}

\subsection{From the Counter Game to BPA Bisimilarity} \label{sub-bisim}

We now reduce the problem of determining the winner in a hit-or-run game to the BPA bisimilarity problem.
To this end we will use the ``gadgets'' of Figure~\ref{fig-gadgets}.
\begin{figure}
\begin{center}
\begin{tabular}{c@{\qquad\qquad}c}
\begin{tikzpicture}[xscale=1.6]
\tikzstyle{mystate} = [minimum height=6mm, minimum width=8mm,rounded corners,inner sep=0,draw];
\tikzstyle{every node} = [mystate];

\node (s) at  (1,3)    {$s$};
\node (s') at (2,3)    {$s'$};
\node (12) at (0,2)    {$u_{12}$};
\node (1'2') at (1,2)  {$u_{1'2'}$};
\node (12') at (2,2)   {$u_{12'}$};
\node (1'2) at (3,2)   {$u_{1'2}$};
\node (t1) at (0,0)    {$t_1$};
\node (t1') at (1,0)   {$t_1'$};
\node (t2) at (2,0)    {$t_2$};
\node (t2') at (3,0)   {$t_2'$};
\draw[->] (s)--(1'2');
\draw[->] (s)--(12);
\draw[->] (s')--(12');
\draw[->] (s')--(1'2);
\draw[->] (1'2')--(t1');
\draw[->] (1'2')--(t2');
\draw[->] (12)--(t1);
\draw[->] (12)--(t2);
\draw[->] (12')--(t1);
\draw[->] (12')--(t2');
\draw[->] (1'2)--(t1');
\draw[->] (1'2)--(t2);
\end{tikzpicture}
&
\begin{tikzpicture}[xscale=1.6]
\tikzstyle{mystate} = [minimum height=6mm, minimum width=8mm,rounded corners,inner sep=0,draw];
\tikzstyle{every node} = [mystate];

\node (s) at  (1,3)  {$s$};
\node (s') at (2,3)  {$s'$};
\node (u1) at (0,2)    {$u_1$};
\node (u1') at (1,2)   {$u_1'$};
\node (u2) at (2,2)    {$u_2$};
\node (u2') at (3,2)   {$u_2'$};
\node (t1) at (0,0)    {$t_1$};
\node (t1') at (1,0)   {$t_1'$};
\node (t2) at (2,0)    {$t_2$};
\node (t2') at (3,0)   {$t_2'$};
\node (bot) at (2.5,1) {$\bot$};
\draw[->] (s)--(u1);
\draw[->] (s)--(u2);
\draw[->] (s')--(u1');
\draw[->] (s')--(u2');
\draw[->] (u2)--(bot);
\draw[->] (u2')--(bot);
\draw[->] (u1)--(t1);
\draw[->] (u1')--(t1');
\draw[->] (u2)--(t2);
\draw[->] (u2')--(t2');
\end{tikzpicture}
%\\ (a) & (b)
\end{tabular}
\end{center}
\caption{(a) Or-gadget \hspace{40mm} (b) And-gadget}
\label{fig-gadgets}
\end{figure}
We have the following lemma, adapted from~\cite{ChenBW12}.
\begin{lemma}[see \cite{ChenBW12}] \label{lem-gadgets}
Consider the states and transitions in Figure~\ref{fig-gadgets} (a) or~(b) as part of an LTS.
The states $s, s'$ may have incoming transitions, the states $t_1,t_1',t_2,t_2',\bot$ may have outgoing transitions (not shown).
All transitions in the figure are labelled with the same action (not shown).
Assume that $\bot \not\sim_1 t$ for all $t \in \{t_1,t_1',t_2,t_2'\}$.
Then we have for the gadgets in Figure~\ref{fig-gadgets} and $\ell \ge 1$:
\begin{itemize}
 \item[(a)]
  for the Or-gadget:  $s \sim_{\ell+2} s'$ if and only $t_1 \sim_\ell t_1'$ or  $t_2 \sim_\ell t_2'$;
 \item[(b)]
  for the And-gadget: $s \sim_{\ell+2} s'$ if and only $t_1 \sim_\ell t_1'$ and $t_2 \sim_\ell t_2'$.
\end{itemize}
\end{lemma}
\begin{proof}
All claims are easy to verify.
As an example, we show for~(a) that $s \sim_{\ell+2} s'$ implies that $t_1 \sim_\ell t_1'$ or $t_2 \sim_\ell t_2'$.
Let $s \sim_{\ell+2} s'$.
Then $u_{12} \sim_{\ell+1} u_{12'}$ or $u_{12} \sim_{\ell+1} u_{1'2}$.
Assume $u_{12} \sim_{\ell+1} u_{12'}$ (the other case is symmetric).
It follows $t_2 \sim_\ell t_2'$ or $t_2 \sim_\ell t_1$.
Similarly, it follows $t_2' \sim_\ell t_2$ or $t_2' \sim_\ell t_1$.
We conclude $t_2 \sim_\ell t_2'$ or $t_2 \sim_\ell t_1 \sim_\ell t_2'$, hence $t_2 \sim_\ell t_2'$.
\qed
\end{proof}

The Or-gadget in Figure~\ref{fig-gadgets}~(a) can be seen as a variant of the ``defender's forcing technique'';
 e.g., the construction of~\cite[Figure~3]{JancarS08} could also be used.
We prefer the gadgets in Figure~\ref{fig-gadgets} as they also work in the probabilistic case, see Section~\ref{sub:probabilistic}.
Now we are ready to prove Theorem~\ref{thm-main}:
\begin{proof}[of Theorem~\ref{thm-main}]
Recall that Proposition~\ref{prop-counter-game} states that the problem of determining the winner in a hit-or-run game is EXPTIME-complete.
We reduce this problem in polynomial time to BPA bisimilarity:
we construct in polynomial time an instance of the BPA bisimilarity problem,
 so that we have bisimilarity if and only if Player~$0$ can win the hit-or-run game.

Let $(S_0, S_1, \mathord{\tran{}}, s_\init, s_\fin, k_\fin)$ be a hit-or-run game.
Let $b \in \N$ such that $2^b > k_\fin$.
W.l.o.g.\ we assume that all states $s \in S$ have exactly two outgoing transitions,
 say $s \tran{\ell(s)_1} t(s)_1$ and $s \tran{\ell(s)_2} t(s)_2$, with $\ell(s)_1, \ell(s)_2 \in \{0, \ldots, 2^b\}$ and $t(s)_1, t(s)_2 \in S \cup \{s_\fin\}$.

Let $\Bin := \{\#_0, \#_1, \ldots, \#_b\}$.
We define a map $\bdec{\cdot} : \Bin^* \to \N$ by $\bdec{\#_{i(1)} \#_{i(2)} \cdots \#_{i(n)}} := \sum_{j=1}^n 2^{i(j)}$.
Conversely, we define a map $\benc{\cdot} : \{0, \ldots, 2^b\} \to \Bin^*$ by $\benc{n} := \#_{i(1)} \#_{i(2)} \cdots \#_{i(z(n))}$
 where $z(n) \in \{0, \ldots, b\}$ and $0 \le i(1) < i(2) < \ldots < i(z(n)) \le b$ and $n = \bdec{\benc{n}}$. %$n = \sum_{j=1}^{j_n} 2^{i(j)}$.
Intuitively, $\benc{n}$ is the sequence of set bits in the (unique) binary representation of~$n$,
 and $z(n)$ is the number of set bits in this representation.

We construct a BPA $(\Gamma,\Sigma,\btran{})$ with an action $a \in \Sigma$.
We write $X \btran{} \alpha$ to mean $X \btran{a} \alpha$.
We include $\Bin \subseteq \Gamma$ and the following rules for $i \in \{0, \ldots, b\}$:
\[
 \#_i \btran{} \#_0 \cdots \#_{i-1}\,.
\]
Note that $\#_0 \btran{} \varepsilon$ where $\varepsilon$ denotes the empty word.
It follows for the generated LTS that $\alpha \tran{} \beta$ with $\alpha, \beta \in \Bin^*$ implies that $\bdec{\beta} = \bdec{\alpha} - 1$,
 i.e., we have effectively implemented a counter that counts down.

We also include a stack symbol $\bot \in \Gamma$, either with no outgoing rules or with a single outgoing rule $\bot \btran{\bar{a}} \bot$
 for an action $\bar{a} \in \Sigma$ with $a \ne \bar{a}$,
 depending on whether we want to avoid a second action or a dead stack symbol.
Observe that for $\alpha, \alpha' \in \Bin^*$ and $\beta, \beta' \in \Gamma^*$
 we have $\alpha \bot \beta \sim \alpha' \bot \beta'$ if and only if $\bdec{\alpha} = \bdec{\alpha'}$.
We include two fresh symbols $s_\fin, s_\fin' \in \Gamma$.
We also include rules $s_\fin \btran{} \varepsilon$ and $s_\fin' \btran{} \benc{k_\fin} \bot$.
By the previous observation, for all $\alpha \in \Bin^*$ we have
\begin{equation} \label{eq-ind-base}
 s_\fin \alpha \bot \sim s_\fin' \alpha \bot \quad \text{if and only if} \quad \bdec{\alpha} = k_\fin\,.
\end{equation}
Let $S := S_0 \cup S_1$, and let $S' := \{s' \mid s \in S\}$ be a copy of~$S$.
We include $S$ and~$S'$ in $\Gamma$.
We now aim at generalizing~\eqref{eq-ind-base} to all $s \in S$ and $\alpha \in \Bin^*$ as follows:
\begin{equation} \label{eq-ind}
 s \alpha \bot \sim s' \alpha \bot \quad \text{if and only if} \quad (s,\bdec{\alpha}) \text{ is winning for Player~$0$.}
\end{equation}
Recall that for each $s \in S$ we have two outgoing transitions $s \tran{\ell(s)_1} t(s)_1$ and $s \tran{\ell(s)_2} t(s)_2$.
For each $s \in S_0$ (or $s \in S_1$, respectively) we include fresh symbols $u(s)_{12}, u(s)_{1'2'}, u(s)_{12'}, u(s)_{1'2} \in \Gamma$
 (or $u(s)_1, u(s)_1', u(s)_2, u(s)_2' \in \Gamma$, respectively).
For each $s \in S_0$ we implement an Or-gadget, see Figure~\ref{fig-gadgets}~(a):
\begin{align*}
 s           &\btran{} u(s)_{1 2 }                  & s'          &\btran{} u(s)_{1 2'} \\
 s           &\btran{} u(s)_{1'2'}                  & s'          &\btran{} u(s)_{1'2 } \\
 u(s)_{1 2 } &\btran{} t(s)_1 \benc{\ell(s)_1}      & u(s)_{1 2'} &\btran{} t(s)_1 \benc{\ell(s)_1} \\
 u(s)_{1 2 } &\btran{} t(s)_2 \benc{\ell(s)_2}      & u(s)_{1 2'} &\btran{} t(s)_2'\benc{\ell(s)_2} \\
 u(s)_{1'2'} &\btran{} t(s)_1'\benc{\ell(s)_1}      & u(s)_{1'2 } &\btran{} t(s)_1'\benc{\ell(s)_1} \\
 u(s)_{1'2'} &\btran{} t(s)_2'\benc{\ell(s)_2}      & u(s)_{1'2 } &\btran{} t(s)_2 \benc{\ell(s)_2}
%\end{align*}
\intertext{For each $s \in S_1$ we implement an And-gadget, see Figure~\ref{fig-gadgets}~(b):}
%\begin{align*}
 s      &\btran{} u(s)_1                  & s'      &\btran{} u(s)_1' \\
 s      &\btran{} u(s)_2                  & s'      &\btran{} u(s)_2' \\
 u(s)_1 &\btran{} t(s)_1 \benc{\ell(s)_1} & u(s)_1' &\btran{} t(s)_1' \benc{\ell(s)_1} \\
 u(s)_2 &\btran{} t(s)_2 \benc{\ell(s)_2} & u(s)_2' &\btran{} t(s)_2' \benc{\ell(s)_2} \\
 u(s)_2 &\btran{} \bot                    & u(s)_2' &\btran{} \bot
\end{align*}
A straightforward induction using Lemma~\ref{lem-gadgets} now establishes~\eqref{eq-ind}.
% for all $\alpha \in \Bin^*$
% that $s \benc{\alpha} \bot \not\sim s' \benc{\alpha} \bot$ holds if and only if configuration $(s,\bdec{\alpha})$ is winning for Player~$1$.
%Hence \eqref{eq-ind} is proved.
Finally, we include two fresh symbols $X,X' \in \Gamma$ and two rules $X \btran{} s_\init \bot$ and $X' \btran{} s_\init' \bot$.
It follows that $X \sim X'$ holds if and only if Player~$0$ can win the hit-or-run game.
This completes the reduction.
\qed
\end{proof}

\section{Remarks}

\subsection{No Dead Symbols and Only One Action}

We remark that Theorem~\ref{thm-main} does not extend to BPAs without dead symbols and with only one action:
In fact, we have the following proposition:
\begin{proposition} \label{prop-not-stronger}
 The bisimilarity problem for BPAs without dead symbols and with only one action is in P.
\end{proposition}
\begin{proof}
 Given a BPA $(\Gamma,\Sigma,\mathord{\btran{}})$ with no dead symbols,
  define the \emph{norm} $|\alpha| \in \N \cup \{\infty\}$ of $\alpha \in \Gamma^*$ as
   the length of the shortest path from~$\alpha$ to~$\varepsilon$ in the generated LTS.
 Note that $|\varepsilon| = 0$ and $|\alpha \beta| = |\alpha| + |\beta|$.
 One can easily compute~$|X|$ for all $X \in \Gamma$ in polynomial time.
 If $|X| < \infty$ holds for all $X \in \Gamma$, we say that the BPA is \emph{normed}.
 The bisimilarity problem for normed BPAs is in~P~\cite{Hirshfeld1996}.

 Let $\Delta = (\Gamma,\{a\},\mathord{\btran{}})$ be a (possibly unnormed) BPA with no dead symbols and only one action~$a$,
  and let $X_\init, X_\init' \in \Gamma$ be two initial states.
 In the following we show how to determine in polynomial time whether $X_\init \sim X_\init'$ holds.
 If $|X_\init| = |X_\init'| = \infty$, then $X_\init \sim X_\init'$;
  in fact, all $\alpha \in \Gamma^*$ with $|\alpha| = \infty$ are bisimilar.
 If exactly one of $|X_\init|, |X_\init'|$ is infinite, then $X_\init \not\sim X_\init'$.
 Hence, assume $|X_\init| < \infty$ and $|X_\init'| < \infty$.
 Define the normed BPA $\Delta_\bullet := (\Gamma_\bullet, \{a,\bar{a}\}, \mathord{\btran{}_\bullet})$ with
  $\Gamma_\bullet := \{X \in \Gamma \mid |X| < \infty\}$,
  and $X \btran{a}_\bullet \alpha$ if $X \btran{a} \alpha$ and $|\alpha| < \infty$,
  and $X \btran{\bar{a}}_\bullet X$ if $X \btran{a} \alpha$ and $|\alpha| = \infty$.
 We claim that we have $X_\init \sim X_\init'$ if and only if $X_\init \sim_\bullet X_\init'$,
  where $\mathord{\sim}$ and~$\mathord{\sim_\bullet}$ mean bisimilar in $\Delta$ and~$\Delta_\bullet$, respectively.
 Since $\Delta_\bullet$ is normed, the latter condition can be decided in polynomial time, as mentioned above.

 It remains to show the claim.
 We show for $\alpha, \alpha' \in \Gamma_\bullet^*$ that we have $\alpha \sim \alpha'$ if and only if $\alpha \sim_\bullet \alpha'$.
 It is easy to verify that $\{(\beta,\beta') \in \Gamma_\bullet^* \times \Gamma_\bullet^* \mid \beta \sim \beta'\}$
  is a bisimulation in~$\Delta_\bullet$, hence $\alpha \sim \alpha'$ implies $\alpha \sim_\bullet \alpha'$.
 Similarly, $\mathord{\sim_\bullet} \cup \{(\beta,\beta') \in \Gamma^* \times \Gamma^* \mid |\beta| = |\beta'| = \infty\}$
  is a bisimulation in~$\Delta$, hence $\alpha \sim_\bullet \alpha'$ implies $\alpha \sim \alpha'$.
\qed
\end{proof}

\subsection{Fully Probabilistic BPAs} \label{sub:probabilistic}

Our main result and its proof generalize to \emph{fully probabilistic BPAs}.

A \emph{probabilistic LTS (pLTS)} $(S,\Sigma,\mathord{\tran{}})$ is like an LTS,
 except that we have $\mathord{\tran{}} \subseteq S \times \Sigma \times \mathcal{D}(S)$,
  where $\mathcal{D}(S)$ denotes the set of probability distributions on~$S$.
For instance, we could have $s \btran{a} d$ with $d(t_1) = 0.7$ and $d(t_2) = 0.3$.
Given a pLTS, a \emph{bisimulation} is an \emph{equivalence} relation $R \subseteq S \times S$ such that
 for all $(s,s') \in R$ and all transitions $s \tran{a} d$
 there is $d'$ with $s' \tran{a} d'$ and $\sum_{s \in E} d(s) = \sum_{s \in E }d'(s)$ holds for every~$E \subseteq S$ that is an equivalence class of~$R$.
As before, two states are bisimilar if a bisimulation relates them, see~\cite{ChenBW12} for more details.

A \emph{probabilistic BPA (pBPA)} $(\Gamma,\Sigma,\mathord{\btran{}})$ is like a BPA,
 except that we have $\mathord{\btran{}} \subseteq \Gamma \times \Sigma \times \mathcal{D}(\Gamma^*)$.
%For instance, we could have $X \btran{a} d$ with $d(XY) = 0.7$ and $d(\varepsilon) = 0.3$.
A pBPA is \emph{fully probabilistic} if for each $X \in \Gamma$ and $a \in \Sigma$
 there is at most one distribution~$d$ with $X \btran{a} d$.
A pBPA induces a pLTS in the same way as a BPA induces an LTS.
The \emph{pBPA bisimilarity problem} is defined as expected,
 i.e., it asks whether two given stack symbols of a given pBPA are bisimilar.
We have the following theorem:
\begin{theorem} \label{thm-prob}
 The pBPA bisimilarity problem is EXPTIME-hard, even if the pBPA is fully probabilistic and has only one action (and a dead stack symbol).
\end{theorem}
\begin{proof}
The proof is completely analogous to the proof of Theorem~\ref{thm-main},
 if the former nondeterministic branching is replaced by uniform probabilistic branching;
  e.g., in the And-gadget of Figure~\ref{fig-gadgets}~(b) we now have $s \tran{a} d$ with $d(u_{1}) = 0.5$ and $d(u_{2}) = 0.5$,
   and $u_1 \tran{a} d'$ with $d'(t_1) = 1$.
Lemma~\ref{lem-gadgets} carries over to the probabilistic case~\cite{ChenBW12},
 and so the construction of Section~\ref{sec-main} establishes the theorem.
\qed
\end{proof}

\subsection{Future Work}

Closing the gap between our EXPTIME lower bound for BPA bisimilarity
 and the 2EXPTIME upper bound from~\cite{Burkart01,JancarBPA13} is an obvious possible target for future research.

\bigskip
\noindent
\textbf{Acknowledgements.}
I would like to thank Javier Esparza, Vojt\v{e}ch Forejt, Christoph Haase, Petr Jan\v{c}ar,
 Ji\v{r}\'{\i} Srba, and James Worrell for valuable hints and discussions,
and anonymous reviewers for their insightful and constructive comments.

%% References
%%
%% Following citation commands can be used in the body text:
%% Usage of \cite is as follows:
%%   \cite{key}         ==>>  [#]
%%   \cite[chap. 2]{key} ==>> [#, chap. 2]
%%

%% References with bibTeX database:

\bibliographystyle{elsarticle-num}%{model1-num-names}
\biboptions{sort&compress}

\bibliography{db}

%% Authors are advised to submit their bibtex database files. They are
%% requested to list a bibtex style file in the manuscript if they do
%% not want to use elsarticle-num.bst.

\end{document}